\documentclass[letterpaper,10 pt, conference]{ieeeconf}  
\usepackage{mymacros}
\usepackage[
backend=biber,
style=ieee,
]{biblatex}
\AtEveryBibitem{%
    \clearfield{doi}
    \clearfield{issn}
}
\DeclareSourcemap{
  \maps[datatype=bibtex]{
    \map{
      \step[fieldset=urldate, null]
    }
  }
}

\IEEEoverridecommandlockouts                        
\usepackage{amsmath} 
  \usepackage{amsthm}
\usepackage{amssymb}  
\usepackage{mathtools}
\usepackage{lipsum}
\usepackage[hidelinks]{hyperref}

\newtheorem{proposition}{Proposition}

\newtheorem{lemma}{Lemma}

\newtheorem{assumption}{Assumption}

\usepackage{graphicx}
\usepackage{caption}
\usepackage{subcaption}

\usepackage{tikz}
\usepackage{standalone}
\usepackage{svg}
\usetikzlibrary{arrows.meta, positioning, calc, backgrounds}

\title{\LARGE \bf
Optimal Functional Incentives for Control:\\The Linear-Quadratic Case with Bilinear Incentives
}

\author{Jonas G. Matt, Saverio Bolognani, Florian Dörfler
\thanks{The authors are with the Automatic Control Laboratory (IfA) at ETH Zürich. Email: \{jmatt, bsaverio, doerfler\}@ethz.ch}%
\thanks{This work was supported by the Swiss Federal Office of Energy via the grant SI/502734 MAESTRO.}
}

\begin{document}

\maketitle
\thispagestyle{empty}
\pagestyle{empty}

\begin{abstract}

We study the design of functional incentive mechanisms for dynamical systems, in which a leader parametrizes a fixed incentive function to motivate a self-interested follower to actuate the system beneficially over an extended horizon, without real-time revision of the incentive.
This stands in contrast to the adaptive paradigm, in which the incentive is itself a continuously updated control variable.
We formalize the problem as a discrete-time bi-level optimal control problem and derive analytical results for the linear-quadratic case with bilinear incentives and a myopic follower.
Specifically, we establish a necessary and sufficient stability condition for the induced closed-loop system, derive a closed-form expression for the gradient of the expected leader cost with respect to the incentive parameter matrix, and obtain a fully closed-form cost expression in the scalar setting.
Based on the latter, explicit characterizations of the optimal incentive parameter are provided in two asymptotic regimes: the infinite-horizon limit and the limit of high follower cost.
For long horizons, the optimal incentive is shown to become independent of the follower's private cost parameter, with direct implications for robust mechanism design under private information.
\end{abstract}

\section{Introduction}
\label{sec:introduction}

The control of modern infrastructure systems such as power grids 
increasingly relies on resources owned and operated by independent agents whose participation must be economically remunerated rather than administratively mandated.
Incentive-based control emerges as a promising paradigm for such settings: a system operator (the \emph{leader}) designs a payment mechanism that incentivizes participating agents (the \emph{followers}) to take actions aligned with system-level objectives.

\subsection*{Adaptive vs.\ Functional Incentive Mechanisms}
\begin{figure}[t]
    \vspace{6pt}
    \centering
    \includegraphics[width=.9\linewidth]{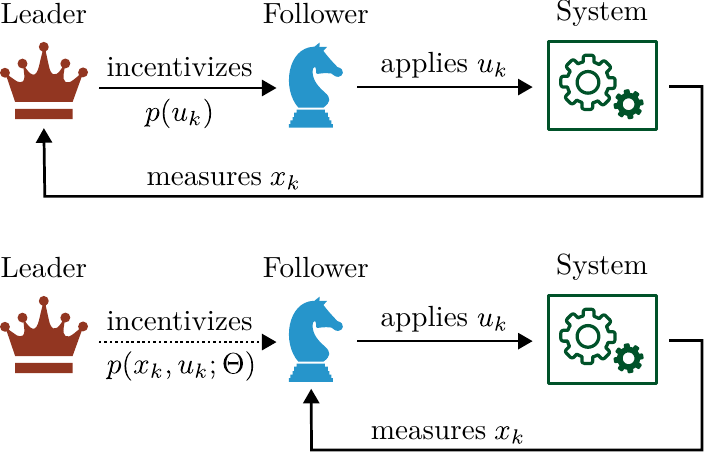}
    \vspace{6pt}
    \caption{
        Adaptive (top) vs. functional (bottom) incentive-based control. In adaptive schemes, the leader must update the incentive signal to elicit new open-loop follower responses. The control loop is effectively closed via the leader. In contrast, functional incentives aim to rationalize beneficial closed-loop control by the follower, on a timescale that is as fast as the follower can interact with the system; the incentive parameters $\Theta$ are fixed in advance and are not part of this fast loop.
    }
    \label{fig:adaptive-vs-functional}
\end{figure}
The predominant approach in the existing literature treats incentives as \emph{adaptive} signals that the leader updates continuously in response to observed states or agent behavior.
Ratliff et al.~\cite{ratliffPerspectiveIncentiveDesign2019,ratliffAdaptiveIncentiveDesign2021} develop online algorithms for adaptive incentive signals and address information asymmetry through simultaneous learning of the follower's decision problem. 
Feedback optimization  methods~\cite{zhouIncentiveBasedOnlineOptimization2018,cavraroFeedbackOptimizationIncentives2024b} adjust adaptive price signals in real time based on system measurements.
Maheshwari et al.~\cite{maheshwariAdaptiveIncentiveDesign2025a} characterize adaptive incentives via the difference between the leader's and follower's marginal costs.
In all of these frameworks, the incentive is a real-time decision variable of the leader that must change as frequently as new follower responses should be elicited.
As one consequence, disturbances can be rejected only as fast as the leader adapts the incentive.

This paper investigates a structurally distinct regime: \emph{functional incentive mechanisms}.
The leader and follower agree in advance to an ``incentive function'', whose parametrization is fixed over an extended time horizon.
The incentive function specifies remuneration as a function of the follower's actions and the current system state but is not updated by the leader as the system evolves.
The follower responds autonomously to changing conditions under this fixed scheme. 
The central design requirement is therefore that desirable behavior for the leader be simultaneously rational for the follower.
In the absence of adaptation, this must be achieved through the structure and parametrization of the incentive function itself.
While the follower's best response to an adaptive price is simply a feedforward decision, the best response to the incentive function can be a feedback law.
This allows to close a direct loop between the system and the follower, as opposed to closing the loop through the leader.
Figure \ref{fig:adaptive-vs-functional} illustrates this fundamental difference.

In practice, functional incentives thus reflect the leader's desire or necessity to \emph{``outsource''} closed-loop control.
The leader may lack the observability, computational resources, or contractual flexibility for real-time intervention, and must instead encode the desired closed-loop behavior into the (offline) incentive design.
A concrete example of the described setup arises in the voltage support program of the Swiss transmission system operator~\cite{karagiannopoulosActiveDistributionGrids2021,jiangVoltageSupportProcurement2025a}.
Participants in this program (distribution system operators and power plants) are compensated for reactive power provision according to a voltage-dependent incentive, rather than through real-time market clearing.
The incentive structure remains fixed over extended periods while the participants respond to real-time grid conditions.
As distributed energy resources proliferate, such mechanisms are expected to become more prevalent across a range of ancillary services.
Accordingly, some functional incentive mechanisms have been proposed for the specific context of power systems \cite{uchidaIncentivizingMarketControl2019,wasaOptimalAgencyContract2019}.
However, the principled design of functional incentive mechanisms for general dynamical systems has overall received little attention in the existing literature.

\subsection*{Other Related Work}

\paragraph{Control-theoretic incentive design}
The intersection of incentive design and control theory was actively examined in the 1980s within the framework of dynamic Stackelberg games.
Ho et al.~\cite{hoControltheoreticViewIncentives1982}, Başar and Ho~\cite{basarAffineIncentiveSchemes1982},
Zheng et al.~\cite{zhengExistenceDerivationOptimal1982}, Saksena et al.~\cite{saksenaOptimalNearOptimalIncentive1983}, and others established the existence of optimal affine incentive strategies for eliciting a prespecified follower response.
Nonlinear extensions were subsequently studied in~\cite{tolwinskiClosedloopStackelbergSolution1981, zhangNonlinearIncentiveStrategy1986,liuOptimalIncentiveStrategy1992, liApproachDiscretetimeIncentive2002}.
A fundamental difference separates these formulations from the present work: they take the desired follower input trajectory as a given design target, computed separately by the leader prior to incentive construction, and typically establish existence or construction results for an incentive achieving that prescribed target.
In the proposed formulation, no such trajectory is prescribed; instead, the leader aims to design a mechanism under which the follower's self-interested optimization produces favorable system behavior across all encountered states, posing the incentive design itself as a (bi-level) optimization problem over the parameter values.
Additionally, the earlier literature generally does not take into account that the leader bears the cost of incentive payments, a budget-coupling that is central to finding the economically optimal incentive.

\paragraph{Principal-agent models and contract design}
The design of contracts that align the objectives of a self-interested agent with those of a principal is a classical topic in economics and management theory~\cite{holmstromAggregationLinearityProvision1987,laffontTheoryIncentivesPrincipalAgent2001}.
Beyond static formulations, a substantial dynamic and continuous-time contract theory literature addresses moral hazard under hidden agent effort, modeling the principal's observation as a stochastic output process and deriving the optimal contract via stochastic control or backward stochastic differential equations~\cite{sannikovContinuoustimeVersionPrincipalagent2008,cvitanicContractTheoryContinuoustime2013}.
In contrast, the present setting features fully observed state and action, deterministic linear dynamics, and a focus on the closed-loop stability and tracking performance of the induced dynamical system, placing the framework closer to control-theoretic incentive design than to hidden-action contracting.

\paragraph{Reverse Stackelberg games}
The game-theoretic structure underlying the proposed formulation, in which the leader's decision variable is a \emph{function} rather than a point action, corresponds to the class of reverse, or inverse, Stackelberg games~\cite{grootReverseStackelbergGames2012}.
Existing optimality results for this class, including necessary and sufficient conditions for optimal affine leader functions~\cite{grootOptimalAffineLeader2016}, are developed for static, single-stage decision problems; the present paper instead embeds the reverse Stackelberg structure within a multi-stage dynamic optimal control problem, yielding closed-loop stability conditions and a gradient-based design procedure over an extended horizon.

\subsection*{Contributions}

The contributions of this paper are as follows.

\begin{itemize}
    \item \textbf{Problem formulation.}
    A general bi-level optimal control framework for functional incentive design is introduced, in which the leader's decision variable is the parameter $\Theta$ of an incentive function governing follower behavior over an extended horizon (Section~\ref{sec:problem-setup}).

    \item \textbf{Linear-quadratic-bilinear analysis.}
    For the special case of linear dynamics, quadratic costs, bilinear incentives, and a myopic follower, a connection to the classical linear quadratic regulator (LQR) problem is made and the following results are established (Section~\ref{sec:lqg}):
    \begin{enumerate}
        \item A necessary and sufficient condition for closed-loop stability.
        \item A closed-form expression for the gradient of the expected leader cost with respect to the incentive parameter matrix, enabling gradient-based optimization in the general multi-dimensional setting.
    \end{enumerate}
    And for the scalar case:
    \begin{enumerate}\setcounter{enumi}{2}
        \item A fully closed-form expression for the expected leader cost.
        \item Explicit characterizations of the optimal incentive parameter in two asymptotic regimes of practical interest: the infinite-horizon limit ($N \to \infty$) and the high follower cost limit ($R \to \infty$), including the observation that the long-horizon optimal incentive is independent of the follower's private cost parameter.
    \end{enumerate}
\end{itemize}
\section{Problem Setup}

We consider a discrete-time bi-level leader-follower (Stackelberg) game in which the leader shapes the behavior of a rational follower through a parametrized incentive function.
Figure~\ref{fig:setup} illustrates the general setup.

Let $x_k \in \mathcal{X} \subseteq \mathbb{R}^n$ denote the system state and $u_k \in \mathcal{U} \subseteq \mathbb{R}^m$ the control input provided by the follower at time step $k$.
The system evolves according to the discrete-time dynamics
\begin{equation}
    x_{k+1} = f(x_k, u_k),
    \label{eq:system-dynamics}
\end{equation}
where $f : \mathcal{X} \times \mathcal{U} \to \mathcal{X}$ is the state transition map.
The leader incurs a stage cost $\ell_L : \mathcal{X} \to \mathbb{R}$, which, for instance, penalizes deviations of the state from a desired reference.
The follower incurs a stage cost $\ell_F : \mathcal{U} \to \mathbb{R}$
on the control input.

The leader does not have direct authority over the input $u_k$.
Instead, she designs an incentive function $p : \mathcal{X} \times \mathcal{U} \times \mathcal{P} \to \mathbb{R}$, drawn from a parametrized family $\{p(\,\cdot\,;\Theta)\}_{\Theta \in \mathcal{P}}$, where $\mathcal{P} \subseteq \mathbb{R}^d$ denotes the parameter space.
The incentive $p(x_k, u_k; \Theta)$ represents the remuneration payment from the leader to the follower at stage $k$. It depends on the current system state $x_k$, the follower's chosen input $u_k$, and the fixed parameters $\Theta$.
The leader seeks to minimize her expected cost over all admissible parameter values:
\begin{equation}
    \min_{\Theta \in \mathcal{P}}\; J(\Theta)
    \;=\;
    \min_{\Theta \in \mathcal{P}}\;
    \E_{x_0} \Bigg[ \sum_{k=0}^{N-1}
        \ell_L(x_k) + p(x_k, u_k;\Theta) \Bigg],
    \label{eq:leader-problem}
\end{equation}
where $N$ denotes the leader's planning horizon, or potentially, the timescale between leader updates.
To make the problem analytically tractable, we make the following assumption for the response of the follower:
\begin{assumption}[Myopic Follower]
\label{ass:myopic}
    The follower optimizes their instantaneous profit at each stage,
    without accounting for future system evolution.
\end{assumption}
This is appropriate when the follower cannot forecast the future state trajectory or strongly discounts future payments, and it yields a tractable, single-level reformulation of the leader's problem.
Relaxing it to a farsighted follower who anticipates how $u_k$ affects future incentive payments would replace the static best response~\eqref{eq:follower-problem} with a dynamic programming problem for the follower, coupling the leader's design problem to the resulting value function; we revisit this extension in Section~\ref{sec:conclusion}.
Under Assumption \ref{ass:myopic}, the follower observes the current state $x_k$ and greedily selects the input that minimizes his own net cost under the prevailing incentive:
\begin{equation}
    u_k = \argmin_{v \in \mathcal{U}}\;
    \ell_F(v) - p(x_k, v;\Theta).
    \label{eq:follower-problem}
\end{equation}

A distinctive feature of the described formulation is the dependence of the incentive $p$ on the system state $x_k$.
This transforms $p$ from a simple price signal into a \emph{feedback} mechanism: even with $\Theta$ fixed, the system is controlled by the follower at fast sampling time.
The leader's optimization can be viewed as a one-time (or infrequent) design decision.
The leader may be unable or unwilling to revise incentives at the same rate at which the follower should react to system changes.
The incentive must therefore not only rationalize a particular response at a given state but establish some notion of optimal control as the rational behavior of the follower across all system conditions.

\label{sec:problem-setup}
\begin{figure}[t]
    \vspace{6pt}
    \centering
    \includegraphics[width=\linewidth]{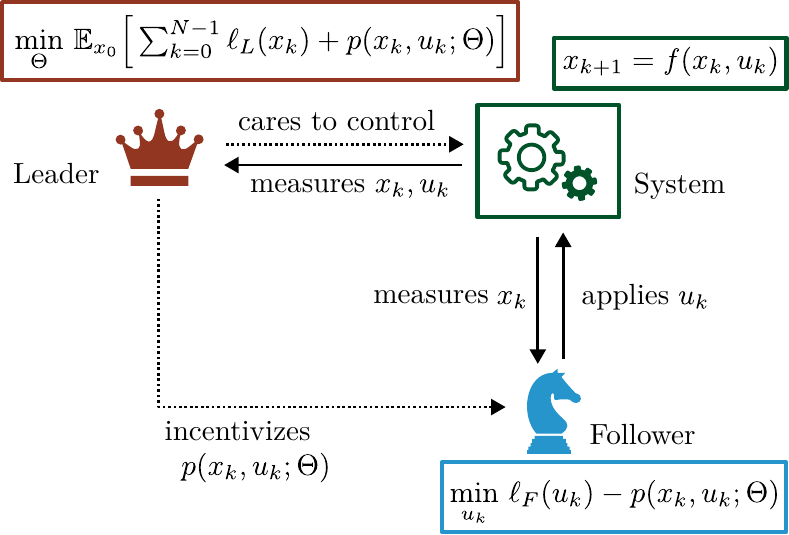}
    \caption{
        Schematic of the functional-incentive-based control setup.
        The leader designs an incentive function parametrized
        by $\Theta$; the follower responds at each stage according to
        their best response, inducing a closed-loop system whose
        behavior depends on the incentive design.
    }
    \label{fig:setup}
\end{figure}
\section{The Linear-Quadratic Case with Bilinear Incentives}
\label{sec:lqg}

In what follows, we derive analytical results for a special but insightful case.
The leader's stage cost is quadratic in the tracking error (the leader cares to track a reference state $\xref\in\mathcal{X}$), $\ell_L(x_k) = (x_k - \xref)^\top Q (x_k - \xref)$, with $Q\in\R^{n\times n}, Q\succ0$;
the follower bears a quadratic input cost, $\ell_F(u_k) = u_k^\top R u_k$, with $R\in\R^{m\times m}, R\succ0$.
The incentive takes the bilinear form
$p(x_k, u_k;\Theta) = (x_k - \xref)\,\Theta\, u_k$, with design parameter $\Theta \in \mathbb{R}^{n \times m}$.
The bilinear form admits a natural interpretation as a state-dependent price per unit of the follower's input and, combined with the quadratic follower cost below, yields a tractable, proportional best response.
The system dynamics are linear, $x_{k+1} = A x_k + B u_k$, with $A \in \mathbb{R}^{n \times n}$ and $B \in \mathbb{R}^{n \times m}$.

The leader's problem therefore reads
\begin{align}
    \min_{\Theta}\; \E_{x_0} \sum_{k=0}^{N-1}&
        \bigl[(x_k - \xref)^\top Q (x_k - \xref)
        + (x_k - \xref)\,\Theta\, u_k\bigr] \nonumber\\
    \text{s.t.}\quad
    & x_{k+1} = A x_k + B u_k, \label{eq:leader-lq}\\
    & u_k = \argmin_{v}\; v^\top R v
        - (x_k - \xref)\,\Theta\, v. \nonumber
\end{align}
We assume that the initial state $x_0\in\mathcal{X}$ is random, and write $\mu_0 := \E[x_0-\xref]$ and $\Sigma_0 := \mathrm{Cov}(x_0-\xref)$ for the mean and covariance of the initial tracking error, and $\bar x_0 := \mu_0+\xref = \E[x_0]$ for the mean of $x_0$ itself.

We can define the \emph{social cost} of the system as the sum of the
leader's and follower's costs:
\begin{align}
    &{\textstyle\sum}_{k=0}^{N-1}\Big[\ell_L(x_k)+p(x_k,u_k;\Theta) + \ell_F(u_k)-p(x_k,u_k;\Theta)\Big] \nonumber\\
    &\quad=\;
    {\textstyle\sum}_{k=0}^{N-1}\Big[(x_k-\xref)^\top Q(x_k-\xref) + u_k^\top R\, u_k\Big].
    \label{eq:social-cost}
\end{align}
The incentive terms cancel, and the social cost reduces to a standard linear-quadratic expression that resembles a linear-quadratic-regulator problem with tracking objective.
This establishes a natural connection between the described problem and classical optimal control theory.
However, the cancellation in~\eqref{eq:social-cost} does not imply that the closed-loop trajectory induced by any incentive function minimizes the social cost.
The outcome of the bi-level game is determined by the follower's best response to the incentive, not by joint minimization of the aggregate objective.
One should therefore expect a gap to the social optimum, analogous to the price of anarchy in game theory.

Under Assumption~\ref{ass:myopic}, the follower's best response conveniently reduces to the proportional control law
\begin{equation}
    u_k = \tfrac{1}{2} R^{-1} \Theta^\top (x_k - \xref),
    \label{eq:follower-input}
\end{equation}
and problem \eqref{eq:leader-lq} collapses to a single-level problem.
Substituting~\eqref{eq:follower-input} into the system dynamics closes the loop between the follower and the system state; the incentive parameter $\Theta$ itself, however, remains fixed over the horizon and is designed open-loop, without revision based on the realized trajectory (cf.\ Fig.~\ref{fig:adaptive-vs-functional}).
We consequently define the closed-loop matrix
\begin{equation}
    A_\Theta = A + \tfrac{1}{2} B R^{-1} \Theta^\top
    \label{eq:A-Theta-def}
\end{equation}
Under the closed-loop dynamics, the tracking error defined as $e_k = x_k - \xref$ evolves according to the affine recursion
\begin{equation}
    e_{k+1} = A_\Theta e_k + (A-I)\xref =: A_\Theta e_k + g,
    \label{eq:error-recursion}
\end{equation}
which unrolls to $e_k = A_\Theta^k e_0 + \sum_{j=0}^{k-1}A_\Theta^j g$.
We deduce the following stability condition for the closed-loop system:
\begin{lemma}[Closed-Loop Stability]
\label{lem:stability}
    Under Assumption~\ref{ass:myopic}, the tracking error $e_k$ converges to the finite steady state $e^\star=(I-A_\Theta)^{-1}g$ for every initial condition $e_0$ if and only if $A_\Theta$ is Schur stable (i.e., all eigenvalues lie strictly inside the unit disk).
\end{lemma}
\begin{proof}
    If $A_\Theta$ is Schur stable, then $A_\Theta^k\to\mathbf 0$ and $\sum_{j=0}^{k-1}A_\Theta^j\to(I-A_\Theta)^{-1}$ as $k\to\infty$, so $e_k\to(I-A_\Theta)^{-1}g=e^\star$ for every $e_0$.
    Conversely, if $A_\Theta$ is not Schur stable, it has an eigenvalue $\mu$ with $|\mu|\geq1$; choosing $e_0$ with a nonzero component along the associated (real) eigenspace, the homogeneous term $A_\Theta^k e_0$ does not vanish as $k\to\infty$, so $e_k$ fails to converge to $e^\star$. Hence convergence for every $e_0$ fails.
\end{proof}
If $A_\Theta$ is Schur stable, the closed-loop system trajectory converges asymptotically to a steady state.
However, because the best response by the follower is purely proportional in the tracking error, we should expect a non-zero steady-state error in general.
The error vanishes only if the reference $\xref$ is an equilibrium of the open-loop system, that is if $(A-I)\xref = \mathbf{0}$.

\subsection{The multi-dimensional case}

Under the closed-loop dynamics, the mean and covariance of the error $e_k = x_k-\xref$ evolve according to
\begin{align}
    \mu_{k+1} &= A_\Theta \mu_k + (A - I)\xref,
        \quad \mu_0 = \E[x_0 - \xref], \nonumber\\
    \Sigma_{k+1} &= A_\Theta \Sigma_k A_\Theta^\top,
        \quad \Sigma_0 = \mathrm{Cov}(x_0 - \xref).
    \label{eq:error-dynamics}
\end{align}
Defining $S = Q+\tfrac{1}{2}\Theta R^{-1} \Theta^\top$, the leader's expected cost can then be written as
\begin{align}
    \label{eq:J-multi-dim}
    J(\Theta)
    = \E\Bigl[\sum_{k=0}^{N-1} e_k^\top S e_k\Bigr]
    = \sum_{k=0}^{N-1} \Tr(S\Sigma_k) + \mu_k^\top S \mu_k
\end{align}
The expected leader cost $J(\Theta)$ is not convex in $\Theta$ in general.
Nonetheless, numerical optimization can be used to find local minima of \eqref{eq:J-multi-dim} for the offline design of the incentive function.
In the following, we derive the gradient of \eqref{eq:J-multi-dim} with respect to $\Theta$, to enable gradient-based optimization.

\begin{proposition}[Gradient of the Leader Cost]
\label{prop:gradient}
    Recursively define the following adjoint variables backward in time:
    \begin{align*}
        \lambda_k &= \bigl(2Q + \Theta R^{-1} \Theta^\top\bigr)\mu_k
            + A_\Theta^\top \lambda_{k+1},
            \quad \lambda_N = \mathbf{0}, \\
        \Lambda_k &= Q + \tfrac{1}{2}\Theta R^{-1}\Theta^\top
            + A_\Theta^\top \Lambda_{k+1} A_\Theta,
            \quad \Lambda_N = \mathbf{0}.
    \end{align*}
    With $\mu_k, \Sigma_k$ computable from the problem parameters as in \eqref{eq:error-dynamics}, the gradient of the expected leader cost with respect to $\Theta$ is
    \begin{align}
        \nabla_\Theta J(\Theta)
        = \sum_{k=0}^{N-1}
        \Bigl[&\bigl(\Sigma_k + \mu_k \mu_k^\top\bigr)\Theta \nonumber\\
        +\; &\bigl(\tfrac{1}{2}
            \mu_k \lambda_{k+1}^\top
            + \Sigma_k A_\Theta^\top \Lambda_{k+1}
        \bigr) B \Bigr] R^{-1}.
        \label{eq:gradient}
    \end{align}
\end{proposition}
\begin{proof}
We compute the first variation of $J(\Theta)$.
Note that $dS = \tfrac{1}{2}(d\Theta R^{-1}\Theta^\top + \Theta R^{-1}d\Theta^\top)$
and, from~\eqref{eq:A-Theta-def}, $dA_\Theta = \tfrac{1}{2}BR^{-1}d\Theta^\top$.
Using $\mu_k^\top dS\,\mu_k = \Tr(dS\,\mu_k\mu_k^\top)$,
\begin{equation*}
    dJ = \sum_{k=0}^{N-1}\Bigl[
        \Tr\bigl[dS(\Sigma_k+\mu_k\mu_k^\top)\bigr]
        + (2S\mu_k)^\top d\mu_k
        + \Tr[S\,d\Sigma_k]
    \Bigr].    
\end{equation*}
Define the adjoint variables as in the proposition statement.
A standard discrete-time adjoint argument (summation by parts applied to the
dynamics of $\mu_k$ and $\Sigma_k$) gives
\begin{align*}
    &\sum_{k=0}^{N-1}\Bigl[(2S\mu_k)^\top d\mu_k + \Tr(S\,d\Sigma_k)\Bigr]
    \\
    = &\sum_{k=0}^{N-1}\Bigl[
        \lambda_{k+1}^\top(dA_\Theta)\mu_k
        + \Tr\bigl(\Lambda_{k+1}\,d\Sigma_{k+1}^{(A_\Theta)}\bigr)
    \Bigr],
\end{align*}
where $d\Sigma_{k+1}^{(A_\Theta)} = (dA_\Theta)\Sigma_k A_\Theta^\top
+ A_\Theta\Sigma_k(dA_\Theta)^\top$.
Using trace identities and symmetry of $R$, the three contributions become:
\begin{align*}
    \Tr\bigl(dS(\Sigma_k+\mu_k\mu_k^\top)\bigr)
        &= \Tr\Bigl((\Sigma_k+\mu_k\mu_k^\top)\Theta R^{-1}d\Theta^\top\Bigr), \\
    \lambda_{k+1}^\top(dA_\Theta)\mu_k
        &= \tfrac{1}{2}\Tr\Bigl(\mu_k\lambda_{k+1}^\top BR^{-1}d\Theta^\top\Bigr), \\
    \Tr\bigl(\Lambda_{k+1}\,d\Sigma_{k+1}^{(A)}\bigr)
        &= \Tr\Bigl(\Sigma_k A_\Theta^\top \Lambda_{k+1} BR^{-1}d\Theta^\top\Bigr).
\end{align*}
Collecting terms gives $dJ = \sum_{k=0}^{N-1}\Tr(G_k^\top\,d\Theta)$ with
\begin{equation*}
    G_k = \Bigl[(\Sigma_k+\mu_k\mu_k^\top)\Theta
        + \bigl(\tfrac{1}{2}\mu_k\lambda_{k+1}^\top
        + \Sigma_k A_\Theta^\top \Lambda_{k+1} \bigr)B\Bigr]R^{-1},
\end{equation*}
and since $dJ = \Tr\bigl((\nabla_\Theta J)^\top d\Theta\bigr)$, the result follows.
\end{proof}

We illustrate the incentive optimization on an example with $n = 2$, $m = 1$,
\begin{equation*}
    A = \begin{bmatrix} 1 & 0.3 \\ 0 & 1 \end{bmatrix},\quad
    B = \begin{bmatrix} 0.5 \\ 1 \end{bmatrix},\quad
    Q = I_2,\quad R = 2,
\end{equation*}
$\bar x_0 = [0,\,0]^\top$, $\xref = [1,\,0]^\top$, and $\Sigma_0 = \mathbf{0}$.
Figure~\ref{fig:multi-dim} shows gradient descent convergence and the closed-loop trajectories under the obtained $\Theta$.
The double-integrator structure of the chosen dynamics yields a vanishing steady-state error.

\begin{figure}[t]
\vspace{6pt}
\centering
\begin{subfigure}{.419\linewidth}
  \centering
  \includegraphics[width=\linewidth]{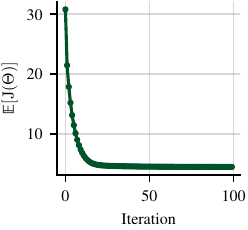}
  \label{fig:multi-dim-convergence}
\end{subfigure}%
\hspace{4pt}
\begin{subfigure}{0.552\linewidth}
  \centering
  \includegraphics[width=\linewidth]{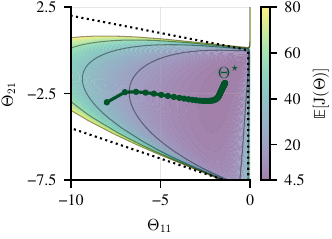}
  \label{fig:multi-dim-convergence-2d}
\end{subfigure}
\vspace{-11pt}
\begin{subfigure}{\linewidth}
  \centering
  \includegraphics[width=\linewidth]{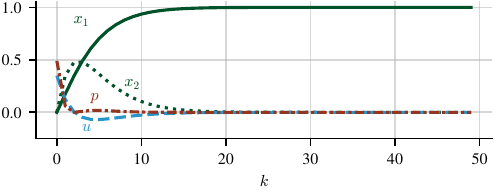}
  \label{fig:multi-dim-optimal}
\end{subfigure}
\caption{
    Top: Convergence of gradient descent on the expected leader cost (Left: Convergence of leader cost; Right: Convergence of incentive parameters).
    Bottom: Closed-loop trajectories of state $x$, input $u$, and incentive payment $p$, under the optimized incentive parameter $\Theta^\star$.
}
\label{fig:multi-dim}
\end{figure}

\subsection{The scalar case}
\label{sec:scalar}

To further illuminate the trade-offs inherent to the design of static incentive mechanisms, we specialize to the scalar setting $n = m = 1$: we first derive a closed-form expression for the leader's cost under general (finite) model parameters, before characterizing the optimal incentive in two asymptotic regimes of practical interest.
For simplicity, we assume $B>0$, although the results are easily extended to the $B<0$ case.
We define the following geometric sums:
\begin{gather}
    \alpha_1 = \frac{1 - A^N}{1 - A},\quad
    \alpha_2 = \frac{1 - A^{2N}}{1 - A^2},\nonumber\\
    \alpha_{\Theta,1} = \frac{1 - A_\Theta^N}{1 - A_\Theta},\quad
    \alpha_{\Theta,2} = \frac{1 - A_\Theta^{2N}}{1 - A_\Theta^2}.
    \label{eq:geometric-sums}
\end{gather}

\begin{proposition}[Closed-Form Leader Cost, Scalar Case]
\label{thm:scalar-cost}
    In the scalar case, the expected leader cost admits the closed form
    \begin{align}
        J(\Theta) = \biggl(&Q + \frac{\Theta^2}{2R}\biggr)   
        \biggl[\,(\Sigma_0 + \mu_0^2)\,\alpha_{\Theta,2} \nonumber\\
        &+ \frac{2(A-1)\xref \mu_0}{1 - A_\Theta}
            \bigl(\alpha_{\Theta,1} - \alpha_{\Theta,2}\bigr) \nonumber\\
        &+ \Bigl(\frac{(A-1)\xref}{1 - A_\Theta}\Bigr)^{\!2}
            \bigl(N - 2\alpha_{\Theta,1} + \alpha_{\Theta,2}\bigr)
        \biggr].
        \label{eq:scalar-cost}
    \end{align}
\end{proposition}
\begin{proof}
    Note that $e_k = A_\Theta^k e_0 + g\,\frac{1 - A_\Theta^k}{1 - A_\Theta}$ by induction.
    Substituting into the cost and pulling the scalar factor outside the sum gives
    \begin{equation*}
        J(\Theta)
        = \Bigl(Q+\frac{\Theta^2}{2R}\Bigr) \sum_{k=0}^{N-1}
            \E\biggl[\biggl(A_\Theta^k e_0
                + \frac{g\,(1-A_\Theta^k)}{1-A_\Theta}\biggr)^{\!2}\biggr].
    \end{equation*}
    Expanding the square and applying linearity of expectation yields
    \begin{align*}
        J(\Theta)
        = \Bigl(Q+\frac{\Theta^2}{2R}\Bigr) \sum_{k=0}^{N-1} &\Bigl[
            \E[e_0^2]\, A_\Theta^{2k} \\
            &+ \frac{2g}{1-A_\Theta}\,\E[e_0]
                \bigl(A_\Theta^{k} - A_\Theta^{2k}\bigr) \\
            &+ \frac{g^2}{(1-A_\Theta)^2}
                \bigl(1 - 2A_\Theta^k + A_\Theta^{2k}\bigr)
        \Bigr].
    \end{align*}
    Substituting $g = (A-1)\xref$, $\E[e_0] = \mu_0$,
    $\E[e_0^2] = \Sigma_0 + \mu_0^2$, and evaluating the geometric
    sums gives~\eqref{eq:scalar-cost}.
\end{proof}
With this analytical expression at hand, we can examine two asymptotic regimes that yield interpretable characterizations of the optimal incentive.

\subsubsection{Infinite-Horizon Limit ($N \to \infty$)}

Assume in the following that $A_\Theta$ is Schur stable, so that the closed-loop trajectory converges.
We exclude the trivial case $\xref=0$ and distinguish the cases $A\neq$ and $A=1$, leading to a non-zero and zero steady-state error, respectively.
First, if $A\neq 1$, the reference is not an equilibrium of the open-loop dynamics and the follower's best response \eqref{eq:follower-input} does not eliminate the steady-state error. Thus, the leader cost \eqref{eq:scalar-cost} diverges as $N \to \infty$.
The steady-state contribution dominates, and the average cost per stage converges to the stage cost in steady state:
\begin{equation}
    \lim_{N \to \infty} \frac{1}{N}\,J(\Theta) = J_\text{ss}(\Theta)
    = \Bigl(Q + \frac{\Theta^2}{2R}\Bigr)
      \Bigl(\frac{(A-1)\xref}{1 - A_\Theta}\Bigr)^{\!2}.
    \label{eq:avg-cost}
\end{equation}
Minimizing \eqref{eq:avg-cost} over $\Theta$ yields the following result.
\begin{proposition}[Scalar Infinite-Horizon Optimal Incentive Parameter, $A\neq1$]
\label{prop:Ninf-with-ss-error}
    If $A\neq1$, the minimizer of the average stage cost in the infinite-horizon limit is equal to the minimizer of the steady-state cost, and it is given by
    \begin{equation}
        \begin{cases}
            \dfrac{QB}{A-1}, &
            \text{if }
            \dfrac{QB^2}{2R(A-1)} + A
            \in (-1,1),\
            \\[8pt]
            \dfrac{-2R(1+A)}{B}, & \text{otherwise.}
        \end{cases}
        \label{eq:theta-opt-Ninf}
    \end{equation}
\end{proposition}
\begin{proof}
    For the closed-loop system to be stable and the steady state to exist, we need
    \begin{equation*}
        |A_\Theta| < 1
        \quad\Longleftrightarrow\quad
        \Theta \in \Bigl(\tfrac{-2R(1+A)}{B},\,\tfrac{2R(1-A)}{B}\Bigr) =: \mathcal{S}.
    \end{equation*}
    We minimize the steady-state cost $J_{\mathrm{ss}}(\Theta)$ over $\mathcal{S}$.
    $J_{\mathrm{ss}}(\Theta)$ is continuously differentiable on $\mathcal{S}$ and the derivative is
    \begin{align*}
    \dfrac{d}{d\Theta} J_{\mathrm{ss}}(\Theta)
    &= (A-1)^2(\xref)^2
    \frac{\Theta(1-A_\Theta)+B\bigl(Q+\frac{\Theta^2}{2R}\bigr)}{R(1-A_\Theta)^3}.
    \end{align*}
    Since $(A-1)\xref \neq 0$, critical points are characterized by the numerator:
    \begin{align*}
    0
    &= \Theta(1-A_\Theta) + B\Bigl(Q+\frac{\Theta^2}{2R}\Bigr) \\
    &= \Theta(1-A) + BQ.
    \end{align*}
    Thus, there exists a unique critical point
    \[
    \Theta^\circ = \frac{BQ}{A-1}.
    \]
    Moreover, $J_{\mathrm{ss}}(\Theta) \to \infty$ as $\Theta \nearrow \frac{2R(1-A)}{B}$,
    so $J_{\mathrm{ss}}$ is not bounded above and the critical point cannot be a local maximum.
    
    Since $(A-1)\xref \neq 0$, the sign of $\frac{d}{d\Theta}J_{\mathrm{ss}}(\Theta)$ on $\mathcal{S}$ is determined solely by the numerator $\Theta(1-A)+BQ$, which is affine in $\Theta$ with slope $1-A \neq 0$.
    Thus, $\Theta^\circ$ must be a local minimum.
    If $\Theta^\circ \in \mathcal{S}$, then $J_{\mathrm{ss}}$ is decreasing before $\Theta^\circ$ and increasing after, until the right boundary of $\mathcal{S}$. Thus, $\Theta^\circ$ is the unique global minimizer on $\mathcal{S}$.
    If $\Theta^\circ \notin \mathcal{S}$, then $J_{\mathrm{ss}}$ is strictly increasing on $\mathcal{S}$ and the minimum must be attained at the left boundary $\Theta = \tfrac{-2R(1+A)}{B}$.
    Finally, the condition $\Theta^\circ \in \mathcal{S}$ is equivalent to $\frac{QB^2}{2R(A-1)} + A \in (-1,1)$.
    Altogether, this yields \eqref{eq:theta-opt-Ninf}.
\end{proof}

The magnitude of the optimal incentive \eqref{eq:theta-opt-Ninf} increases proportionally with $B$ (follower has larger influence on system dynamics) and $Q$ (tracking error is more strongly penalized).
Notably, the optimal incentive does not depend on $R$, except when $R$ determines the closed-loop stability.
This invariance is a consequence of a balancing of two competing effects: larger $\Theta$ elicits greater follower effort (reducing steady-state error) but also increases marginal incentive payments.
As $R$ scales, both effects scale proportionally, leaving the optimum unchanged.
This insight can be useful since $R$ encodes the follower's private cost of effort which is information the leader typically can only estimate.
Proposition \ref{prop:Ninf-with-ss-error} reveals that the optimal long-horizon incentive can be designed without requiring knowledge of $R$.

Now, if $A=1$, the steady-state error is zero and the leader cost converges.
The optimal $\Theta$ thus minimizes the cost of the transient.
The leader cost simplifies to
\begin{equation}
\label{eq:J_for_Ninf_no-ss-error}
\lim_{N \to \infty} J(\Theta) =  \Bigl(Q + \frac{\Theta^2}{2R}\Bigr)
        \frac{\Sigma_0 + \mu_0^2}{1-A_\Theta^2}.
\end{equation}
\begin{proposition}[Scalar Infinite-Horizon Optimal Incentive Parameter, $A=1$]
\label{prop:Ninf-no-ss-error}
    If $A=1$, the minimizer of the infinite-horizon cost is
    \begin{equation}
        \dfrac{BQ}{2} - \sqrt{\frac{B^2Q^2}{4} + 2QR}\ .
        \label{eq:theta-opt-Ninf-no-ss-error}
    \end{equation}
\end{proposition}
\begin{proof}
    Because of $A = 1$, we have $A_\Theta = 1 + B\Theta/2R$, so stability requires $|A_\Theta| < 1$, i.e.\ $\Theta \in \mathcal{S} = \bigl(\frac{-4R}{B}, 0\bigr)$.
    The first order condition reads
    \begin{equation*}
        0 = \frac{\Theta(1-A_\Theta^2)+B\bigl(Q+\frac{\Theta^2}{2R}\bigr)\bigl(A+\frac{B\Theta}{2R}\bigr)}{R(1-A_\Theta^2)^2}
    \end{equation*}
    which, upon substituting $A = 1$, simplifies to the quadratic
    \begin{equation*}
        \Theta^2 - BQ\,\Theta - 2QR = 0.
    \end{equation*}
    The two roots are $\frac{BQ}{2} \pm \sqrt{\frac{B^2Q^2}{4} + 2QR}$.
    Since $Q, R > 0$, the square root strictly exceeds $\frac{BQ}{2}$, so the positive root lies outside $\mathcal{S}$. The negative root
    \begin{equation*}
        \Theta^\circ = \frac{BQ}{2} - \sqrt{\frac{B^2Q^2}{4} + 2QR}
    \end{equation*}
    lies in $\mathcal{S}$: it is negative by the above, and is larger than $-4R/B$ since $\sqrt{\frac{B^2Q^2}{4}+2QR} < \frac{BQ}{2} + \frac{4R}{B}$, which holds after squaring (using $B, Q, R > 0$). Hence $\Theta^\circ$ is the unique critical point in $\mathcal{S}$. Since the cost diverges as $\Theta \nearrow 0$ and is finite at the left boundary, the same monotonicity argument as in Proposition~\ref{prop:Ninf-with-ss-error} implies $\Theta^\circ$ is the global minimizer on $\mathcal{S}$, yielding \eqref{eq:theta-opt-Ninf-no-ss-error}.
\end{proof}
Because of the zero steady-state error the infinite-horizon leader cost depends on the transient and unlike before, is not invariant to any problem parameter.
For large $B$, the optimal $\Theta$ converges to $0$. For increasing $Q$, the optimal incentive decreases but converges to the finite value $-2R/B$. For increasing $R$, the optimal incentive decreases without bounds like $-\sqrt{R}$.

\subsubsection{High Follower Cost ($R \to \infty$)}

We next consider the regime in which the follower's willingness to provide control effort diminishes because their cost for inputs grows.

\begin{proposition}[Optimal Incentive, $R \to \infty$]
\label{prop:Rinf}
    The limiting minimizer
    $\Theta^\star_{R\to\infty} = \lim_{R\to\infty}
    \argmin_{\Theta} J(\Theta)$
    is finite and given by
    \begin{equation}
        \Theta^\star_{R\to\infty} = -\frac{QB}{2}\,\frac{\Gamma'(A)}{\Gamma(A)},
    \end{equation}
    with
    \begin{align*}
        \Gamma(A) &= (\Sigma_0 + \bar x_0^2)\,\alpha_2 - 2\xref \bar x_0\,\alpha_1 + N\,(\xref)^2\\
        \Gamma'(A) &= 
            (\Sigma_0 + \bar x_0^2)\,\frac{d\alpha_2}{dA}
            - 2\xref \bar x_0\,\frac{d\alpha_1}{dA} \\
            &\quad + \frac{2\xref}{1-A}
            \Bigl[\xref(N - \alpha_1)
            - \bar x_0(\alpha_1 - \alpha_2)\Bigr],
    \end{align*}
\end{proposition}
\begin{proof}
Since $J>0$, minimizing $J(\Theta)=\bigl(Q+\frac{\Theta^2}{2R}\bigr)\Gamma(A_\Theta)$ over $\Theta$ is equivalent to minimizing $\log J$. Taylor-expanding around $A_\Theta=A$ for large $R$ and using $\log(1+x)=x+O(x^2)$ gives
\begin{align*}
    \log J(\Theta)
    = &\log Q + \log\Gamma(A)\, + \\
    &\frac{1}{R}\left[
        \frac{\Theta^2}{2Q}
        + \frac{B\,\Gamma'(A)}{2\,\Gamma(A)}\,\Theta
      \right] + O(R^{-2}).    
\end{align*}
Since the $O(1)$ term does not depend on $\Theta$ and $R^{-1}>0$, the minimizer is asymptotically determined by the quadratic in brackets, whose unique minimum is $\Theta^\star_{R\to\infty} = -\frac{QB}{2}\frac{\Gamma'(A)}{\Gamma(A)}$.

It remains to verify the formulas for $\Gamma(A)$ and $\Gamma'(A)$.
Write $c(a):=\frac{(A-1)\xref}{1-a}$, so that
\begin{align*}
    \Gamma(a) = &(\Sigma_0+\mu_0^2)\,\alpha_{a,2}
              + 2c(a)\mu_0(\alpha_{a,1}-\alpha_{a,2})\\
              &+ c(a)^2(N-2\alpha_{a,1}+\alpha_{a,2}),    
\end{align*}
where $\alpha_{a,1},\alpha_{a,2}$ are functions of $a$.
Evaluating at $a=A$ gives $c(A)=-\xref$, $\alpha_{a,1}=\alpha_1$, and $\alpha_{a,2}=\alpha_2$, which yields the stated $\Gamma(A)$.
Differentiating with respect to $a$ via the product rule yields ($a$ arguments omitted for brevity):
\begin{align*}
    \Gamma'(A)
    &= (\Sigma_0+\mu_0^2)\,\alpha_2'
     + 2c'\mu_0(\alpha_1-\alpha_2)
     + 2c\mu_0(\alpha_1'-\alpha_2') \\
    &\quad + 2cc'(N-2\alpha_1+\alpha_2)
     + c^2(-2\alpha_1'+\alpha_2').
\end{align*}
Evaluating at $a=A$, using $c(A)=-\xref$ and $c'(A)=-\frac{\xref}{1-A}$, and collecting coefficients yields the stated $\Gamma'(A)$.
\end{proof}
The finiteness of the limiting minimizer implies that increasingly ``expensive'' followers should not be incentivized without bound.
In settings with asymmetric information where the follower’s cost parameter can only be estimated, this result provides an additional guardrail for designing the incentive function.

Finally, we present another example to numerically corroborate the results for the scalar case.
We choose $N = 10$, $A = 0.4$, $B = 1$, $Q = 1$, $R = 1$, $\xref = 1$, $\bar x_0 = 0$, and $\Sigma_0 = 0.1$; note that this corresponds to a nonzero error mean $\mu_0 = \bar x_0 - \xref = -1$.
Since $(A-1)\xref\neq0$, the discussion following Lemma~\ref{lem:stability} applies: even under the cost-optimal incentive, the state converges to a steady state distinct from $\xref$, as the leader's cost also accounts for incentive payments and therefore does not warrant fully eliminating the tracking error.
Figure~\ref{fig:scalar} and its caption illustrate the results.

\begin{figure}[t]
\centering
\begin{subfigure}{\linewidth}
  \centering
  \caption{}
  \includegraphics[width=\linewidth]{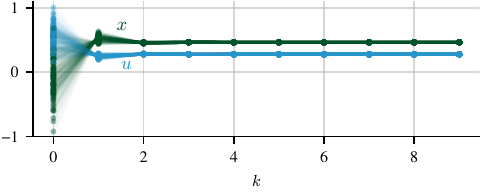}
  \label{fig:scalar-trajectories}
  \vspace{-15pt}
\end{subfigure}
\begin{subfigure}{.47\linewidth}
  \centering
  \caption{}
  \includegraphics[width=\linewidth]{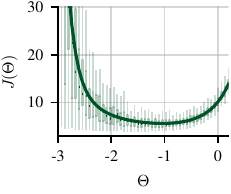}
  \label{fig:scalar-leader-cost}
  \vspace{-10pt}
\end{subfigure}%
\hspace{5pt}
\begin{subfigure}{.48\linewidth}
  \centering
  \caption{}
  \includegraphics[width=\linewidth]{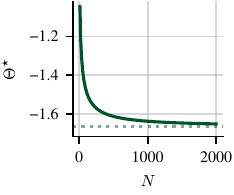}
  \label{fig:scalar-N-limit}
  \vspace{-10pt}
\end{subfigure}
\begin{subfigure}{.5\linewidth}
  \centering
  \caption{}
  \includegraphics[width=\linewidth]{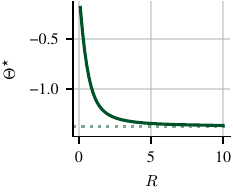}
  \label{fig:scalar-R-limit}
  \vspace{-10pt}
\end{subfigure}%
\hspace{5pt}
\begin{subfigure}{.45\linewidth}
  \centering
  \caption{}
  \includegraphics[width=\linewidth,trim={.4cm 0 0 0},clip]{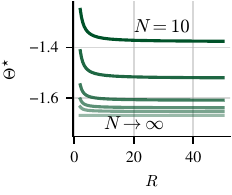}
  \label{fig:scalar-R-limit-different N}
  \vspace{-10pt}
\end{subfigure}
\caption{
    Scalar case analysis.
    (a) State and input trajectories under optimal incentive for 100 sampled initial conditions; the state converges to a steady state distinct from $\xref$ according to Lemma~\ref{lem:stability}.
    (b) Expected leader cost as function of $\Theta$.
    (c)--(d) Behavior of the optimal incentive parameter as functions
    of the horizon and follower cost, respectively, with
    asymptotes consistent with Propositions~\ref{prop:Ninf-with-ss-error} and~\ref{prop:Rinf}.
    (e) Dependence of optimal $\Theta$ on $R$ for increasing horizon lengths $N$. The dependence on $R$ vanishes as $N\to\infty$.
}
\label{fig:scalar}
\end{figure}

\section{Conclusion}
\label{sec:conclusion}

This paper has introduced a bi-level optimal control framework for the design of incentive mechanisms, in which a leader parametrizes a fixed incentive function to elicit favorable closed-loop behavior from a self-interested, myopic follower. 
The framework is motivated by practical settings, such as ancillary service procurement in power systems, where incentive structures are codified in long-term contracts and cannot be adapted at the timescale of system dynamics and disturbances.

For the linear-quadratic case with bilinear incentives, we have derived tractable expressions for the expected cost of the leader and its gradient, enabling locally optimal design of the incentive function by the leader.
Analysis of two asymptotic regimes yields explicit optimal incentives and provides theoretical guardrails for designing the incentive under asymmetric information.
In particular, the optimal incentive in the infinite-horizon limit is independent of the follower's cost parameter, and the optimal incentive for extremely hesitant followers is finite.

Several directions are promising for future work.

\paragraph{Multi-follower settings and collusion}
The present formulation considers a single follower.
An important generalization involves multiple followers who interact strategically at the lower level, introducing a game among followers in response to a common incentive.
This raises questions of equilibrium selection and the collusion-proofness of the resulting mechanism.

\paragraph{Strategic followers}
The myopic best-response assumption is central to the analytical tractability of the present results.
Relaxing this assumption to allow followers that have knowledge of the system dynamics and optimize over a future horizon rather than greedily at each stage, changes the follower's response and the structure of the leader's design problem.
Characterizing optimal static incentives for strategic followers is thus a relevant extension.

\paragraph{Incentive function design}
A limitation of the bilinear incentive studied here is the presence of a nonzero steady-state tracking error under the follower's best response.
Incentive designs that eliminate this offset,
for instance by rationalizing integral control,
are a promising avenue toward improved closed-loop performance within the static incentive paradigm.

\printbibliography

\end{document}